\documentclass[smallextended,envcountsect]{svjour3}
\smartqed
\usepackage{graphicx}
\usepackage{amsfonts}
\usepackage{amsmath}
\usepackage{amssymb}
\usepackage{caption}
\usepackage{setspace}
\usepackage{color}
\usepackage[normalem]{ulem}
\usepackage[colorlinks,linkcolor=blue,citecolor=blue]{hyperref}

\usepackage{ifpdf}
\usepackage{color}
\definecolor{webgreen}{rgb}{0,.5,0}
\definecolor{webbrown}{rgb}{.8,0,0}
\definecolor{emphcolor}{rgb}{0.95,0.95,0.95}

\journalname{JOTA}

\newcommand{\lev}{L\'{e}vy }

\newcommand{\cadlag}{c\`{a}dl\`{a}g }

\newcommand {\p}{\mathbb{P}}
\newcommand {\R}{\mathbb{R}}
\newcommand {\E}{\mathbb{E}}
\newcommand{\e}{\mathbb{E}}
\newcommand {\ud}{\textup{d}}
\newcommand{\diff}{{\rm d}}

\newtheorem{assump}{Assumption}
\allowdisplaybreaks
\begin{document}

\title{On the Bail-Out Optimal Dividend Problem}

\author{Jos\'e-Luis P\'erez \and Kazutoshi Yamazaki   \and  Xiang Yu}

\institute{Jos\'e-Luis P\'erez \at
             Centro de Investigaci\'on en Matem\'aticas \\
              Guanajuato, Mexico\\
              jluis.garmendia@cimat.mx
           \and
           Kazutoshi Yamazaki  \at
              Kansai University \\
              Osaka, Japan\\
              kyamazak@kansai-u.ac.jp
                         \and
           Xiang Yu (Corresponding author) \at
              The Hong Kong Polytechnic University\\
              Kowloon, Hong Kong\\
              xiang.yu@polyu.edu.hk
}

\date{Received: date / Accepted: date}

\maketitle

\begin{abstract}
This paper studies the optimal dividend problem with capital injection under the constraint that the cumulative dividend strategy is absolutely continuous. We consider an open problem of the general spectrally negative case and derive the optimal solution explicitly using the fluctuation identities of the refracted-reflected \lev process. The optimal strategy as well as the value function are concisely written in terms of the scale function.  Numerical results are also provided to confirm the analytical conclusions.
\end{abstract}


\section{Introduction}

In the bail-out model of de Finetti's optimal dividend problem, one wants to maximize the total expected dividends minus the costs of capital injection under the constraint that the surplus must be kept non-negative uniformly in time. Typically, a spectrally negative \lev process (a \lev process with only downward jumps) is used to model the underlying surplus process of an insurance company that increases because of premiums and decreases by insurance payments. Avram et al. \cite{1.} showed that it is optimal to reflect from below at zero and also from above at a suitably chosen threshold.

We investigate an extension with the absolutely continuous constraint on dividend strategies. To be precise, the cumulative dividend process must be absolutely continuous with respect to the Lebesgue measure  with its density bounded by a given constant. This problem (without bail-out) has been previously considered by \cite{2.} and \cite{3.} in the diffusive case, and \cite{4.} for the case of Cr\'amer-Lundberg processes with exponential jumps. In the classical setting, the set of admissible strategies is too general and counterintuitive in the context of insurance. Consequently, there have been several attempts to restrict the solution to more realistic strategies. The absolutely continuous condition is one way of achieving this goal without losing analytical tractability.

Regarding the version with both bail-out and the absolutely continuous condition, the dual case (the spectrally positive \lev case) has recently been solved by \cite{5.}. In this paper, we further consider the spectrally negative case for the underlying process. This can also be seen as the bail-out version of \cite{6.}, where they incorporated the absolutely continuous constraint, however, without capital injections.

Our ultimate aim is to verify the conjecture on the optimality of a \emph{refraction-reflection strategy} that reflects the surplus from below at zero in the classical sense and refracts the process (decreases the drift) at a suitably chosen threshold. The resulting controlled surplus process becomes the so-called refracted-reflected \lev process recently studied in \cite{7.}. Indeed, many interesting probabilistic properties of the refracted-reflected \lev process have been developed in \cite{7.}. However, as an important application to the optimal dividend problem with capital injection, it is still an open problem whether the optimal control for the spectrally negative case fits this type of refraction-reflection. This paper fills the gap and provides the closed-form choice of the threshold.

As is commonly used in the related literature, we adopt the scale function and the fluctuation identities so as to follow efficiently the ``guess-and-verify" procedure described below.
\begin{itemize}
\item[(1)] By focusing on the set of refraction-reflection strategies, we select a judicious candidate strategy via the \emph{smooth fit} principle.  In particular, we choose the threshold value such that the corresponding net present value (NPV) becomes continuously (resp.\ twice continuously) differentiable at the threshold for the case of bounded (resp.\ unbounded) variation.

\item[(2)] The optimality of the selected strategy is then confirmed by verifying the \emph{variational inequalities} that require the computation of the generators and certain slope conditions of the value functions.
\end{itemize}
In general, in the optimal dividend problem and its extensions, the verification of optimality is significantly more challenging for the spectrally negative case than the dual case. The difficulty typically lies in the required proof of the properties of the candidate value function above the barrier/threshold that separates the waiting and controlling regions. Intuitively speaking, this is difficult because, with negative jumps, the surplus can jump below from the controlling region to the waiting region as well as directly to the reflection region below the zero boundary, where the forms of the value function change. It is demonstrated in the literature that the optimality can fail by the choice of the \lev measure (see, e.g., \cite{6.} assumes the completely monotone \lev density). However, in the dual model, it is usually not necessary to assume any property on the \lev measure  (see \cite{8.,9.,10.,11.}).

Mathematically speaking, in our problem, the major challenge is to show the slope condition above the selected threshold such that the slope is bounded uniformly by $1$. Nonetheless, we show that the optimality holds for a general spectrally negative \lev case.  To this end, we use our observation that the slope of the candidate value function coincides with the Laplace transform of the ruin time of the refracted \lev process of \cite{12.}, which is monotone in the starting value. Other required computations such as generators and the slopes below the threshold can be performed efficiently by taking advantage of the analytical properties of the scale function.

The rest of the paper is organized as follows. In Section \ref{section_preliminary}, we review the spectrally negative \lev process and give the precise formulation of the bail-out optimal dividend control problem with the absolutely continuous condition. Section \ref{section_strategies} defines the refraction-reflection strategy and formulates the corresponding NPV of dividends minus capital injection using the scale function. Section \ref{section_threshold} provides the conjectured candidate threshold and Section \ref{section_verification} proves the optimality of the selected strategy. Some numerical examples are presented in Section \ref{section_numerics}. At last, we give our conclusions in Section \ref{section_conc}.

\section{Preliminaries}\label{section_preliminary}
\subsection{Spectrally Negative \lev Processes}
In this paper, we consider a spectrally negative L\'evy process $X$. For $x\in \R$, we denote the law of $X$ when it starts at $x$ by $\p_x$ and refer to it as $\p$ instead of $\p_0$ for convenience. $\e_x$ and $\e$ are the associated expectation operators.

Its Laplace exponent $\psi(\theta):[0,\infty[ \to \R$ is defined by ${\rm e}^{\psi(\theta)t}:=\e\big[{\rm e}^{\theta X_t}\big]$ for $t, \theta\ge 0$ with the \emph{L\'evy-Khintchine formula}
\begin{equation}
\psi(\theta):=\gamma\theta+\frac{\sigma^2}{2}\theta^2+\int_{]-\infty,0[}\big({\rm e}^{\theta z}-1-\theta  z \mathbf{1}_{\{ z>-1\}}\big)\Pi(\ud  z), \quad \theta \geq 0,\notag
\end{equation}
where $\gamma\in \R$, $\sigma\ge 0$, and $\Pi$ is a measure on $]-\infty,0[$ called the L\'evy measure of $X$ that satisfies
$\int_{]-\infty,0[}(1\land  z^2)\Pi(\ud  z)<\infty$.

It is well-known that $X$ has paths of bounded variation if and only if $\sigma=0$ and $\int_{]-1, 0[}  |z|\Pi(\mathrm{d}  z)$ is finite. In this case, $X=\{X_t=ct-S_t, t\geq 0\}$, where $c:=\gamma-\int_{]-1,0[}  z \Pi(\mathrm{d}  z)$ and $\{S_t; t\geq0\}$ is a driftless subordinator. Note that necessarily $c>0$, as we have ruled out the case that $X$ has monotone paths. Its Laplace exponent is given by $\psi(\theta) = c \theta+\int_{]-\infty,0[}\big( {\rm e}^{\theta  z}-1\big)\Pi(\ud  z),$ for  $\theta \geq 0$.

\subsection{Bail-Out Optimal Dividend with the Absolutely Continuous Condition}
A strategy is a pair $\pi := \left( L_t^{\pi}, R_t^{\pi}; t \geq 0 \right)$ of nondecreasing, right-continuous, and adapted processes (with respect to the filtration generated by $X$) starting at zero, where $L^{\pi}$ is the cumulative amount of dividends and $R^{\pi}$ is that of the injected capital. With $V_{0-}^\pi := x$, and, $V_t^\pi := X_t - L_t^\pi + R_t^\pi$, $t \geq 0$, it is required that $V_t^\pi \geq 0$ a.s.\ uniformly in $t$.  In addition,  with $\delta > 0$ fixed, $L^\pi$ is required to be absolutely continuous with respect to the Lebesgue measure of the form $L_t^\pi = \int_0^t \ell^\pi_s \diff s$, $t \geq 0$, with $\ell^\pi$ restricted to take values in $[0,\delta]$ uniformly in time. As for $R^\pi$, it is assumed that $\int_{[0, \infty[} e^{-qt} \diff R_t^{\pi} < \infty$, a.s.

Assuming that $\beta > 1$ is the cost per unit injected capital and $q > 0$ is the discount factor, the expected NPV of dividends minus the costs of capital injection under a strategy $\pi$ becomes
\begin{align*}
v_{\pi} (x) := \mathbb{E}_x \left( \int_0^\infty e^{-q t} \ell_t^\pi  \diff t - \beta \int_{[0, \infty[} e^{-q t} \diff R_t^{\pi}\right), \quad x \in \mathbb{R}.
\end{align*}
The corresponding stochastic control problem is defined by
\begin{equation}\label{control:value}
v(x):=\sup_{\pi \in \mathcal{A}}v_{\pi}(x), \quad x \in \mathbb{R},
\end{equation}
where $\mathcal{A}$ is the set of all admissible strategies that satisfy the constraints described above.

Throughout the paper, to exclude the trivial case, we consider the next assumption.
\begin{assump} \label{assump_mean_finite}
We assume that $\E X_1 = \psi'(0+) > -\infty$.
\end{assump}

Moreover, as being commonly imposed in the literature (see \cite{6.}), the next assumption is made so that the process $Y:=\{Y_t:= X_t - \delta t,t\geq0\}$ does not have monotone paths.

\begin{assump}
For the case of bounded variation, let $c > \delta$.
\end{assump}

\section{Refraction-Reflection Strategies} \label{section_strategies}

Our objective is to show the optimality of the refraction-reflection strategy $\pi^{b} = (L^{0,b}, R^{0,b}, t \geq 0)$, with a suitable refraction level $b \geq 0$. Namely, dividends are paid at the maximal rate $\delta$ whenever the surplus process is above the pre-specified threshold $b$ while it is pushed upward by capital injection whenever it attempts to downcross zero. The resulting surplus process
\begin{align*}
U_t^{0,b} := X_t - L^{0,b}_t + R_t^{0,b}
\end{align*}
becomes the standard refracted-reflected \lev process of \cite{7.}.

In terms of the optimal control theory, we recognize that the candidate dividend strategy is of the bang-bang type, i.e., dividends should either be paid out at the maximum rate $\delta$ or at the rate $0$. On the other hand, the capital injection strategy, which is the reflection control, fits into the singular control framework. To wit, we can explicitly write the described cumulative dividend control as
$L_t^{0,b} = \int_0^t \delta 1_{\{ U_s^{0,b} > b \}} \diff s$,
and \emph{for the case of bounded variation} we can write the candidate capital injection $R_t^{0,b} = \sum_{0 \leq s \leq t} |U_{s-}^{0,b}+\triangle X_s| 1_{\{U_{s-}^{0,b}+\triangle X_s < 0\}}$. Here, we define $\triangle \xi_t:=\xi_t-\xi_{t-}$, $t\geq 0$, for any \cadlag process $\xi$. For a formal construction of this process, we refer the reader to \cite{7.}.

Clearly, each aforementioned refraction-reflection strategy $\pi^{b}$ is admissible for any $b \geq 0$. We denote the corresponding expected NPV by
\begin{align} \label{v_pi}
v_b (x) := \mathbb{E}_x \left( \int_0^\infty e^{-q t} \diff L_t^{0,b} - \beta \int_{[0, \infty[} e^{-q t} \diff R_t^{0,b}\right), \quad x \in \mathbb{R}.
\end{align}

In order to express \eqref{v_pi}, we apply the fluctuation identities. Following the same notations as in \cite{12.}, we call $W^{(q)}$ and $\mathbb{W}^{(q)}$ the scale functions of $X$ and $Y$, respectively.  These are the mappings from $\R$ to $[0, \infty[$ that take the value zero on the negative half-line, while on the positive half-line, they are strictly increasing functions that are defined by their Laplace transforms
\begin{align} \label{scale_function_laplace}
\int_0^\infty  \mathrm{e}^{-\theta x} W^{(q)}(x) \diff x &= \frac 1 {\psi(\theta)-q}, \quad \theta > \Phi(q), \\
\int_0^\infty  \mathrm{e}^{-\theta x} \mathbb{W}^{(q)}(x) \diff x &= \frac 1 {\psi_Y(\theta) -q}, \quad \theta > \varphi(q).
\end{align}
Here $\psi_Y(\theta) := \psi(\theta) - \delta \theta$, $\theta \geq 0$, is the Laplace exponent for $Y$ and
\begin{align}
\begin{split}
\Phi(q) := \sup \{ \lambda \geq 0: \psi(\lambda) = q\} \quad \textrm{and} \quad \varphi(q) := \sup \{ \lambda \geq 0: \psi_Y(\lambda) = q\}. \notag
\end{split}
\end{align}

We also define, for $x \in \R$,
\begin{align*}
\overline{W}^{(q)}(x) &:=  \int_0^x W^{(q)}(y) \diff y, \quad
Z^{(q)}(x) := 1 + q \overline{W}^{(q)}(x),  \\
\overline{Z}^{(q)}(x) &:= \int_0^x Z^{(q)} (z) \diff z = x + q \int_0^x \int_0^z W^{(q)} (w) \diff w \diff z.
\end{align*}
Noting that $W^{(q)}(x) = 0$ for $-\infty < x < 0$, we have
\begin{align}
\overline{W}^{(q)}(x) = 0, \quad Z^{(q)}(x) = 1  \quad \textrm{and} \quad \overline{Z}^{(q)}(x) = x, \quad x \leq 0.  \label{z_below_zero}
\end{align}
Analogously, we define $\overline{\mathbb{W}}^{(q)}$, $\mathbb{Z}^{(q)}$ and $\overline{\mathbb{Z}}^{(q)}$ for $Y$.  
%
From computations in \cite{7.}, we already know
\begin{align}\label{RLqp}
&\delta \int_0^x\mathbb{W}^{(q)}(x-y) W^{(q)}(y) \diff y=\overline{\mathbb{W}}^{(q)}(x)-\overline{W}^{(q)}(x), \quad \\
&\delta \int_0^x\mathbb{W}^{(q)}(x-y) Z^{(q)}(y)\diff y=\overline{\mathbb{Z}}^{(q)}(x)-\overline{Z}^{(q)}(x)+\delta\overline{\mathbb{W}}^{(q)}(x). \label{RLqp2}
\end{align}

\begin{remark} \label{remark_scale_function_properties}
\begin{enumerate}
\item[(i)] $W^{(q)}$ and $\mathbb{W}^{(q)}$ are differentiable almost everywhere. In particular, if $X$ is of unbounded variation or the \lev measure is atomless, it is known that $W^{(q)}$ and $\mathbb{W}^{(q)}$ are $C^1(\R \backslash \{0\})$; see Theorem 3 of \cite{13.}.
\item[(ii)] As $x \downarrow 0$, by Lemma 3.1 of \cite{14.}, we have
\begin{align*}
\begin{split}
W^{(q)} (0) &= \left\{ \begin{array}{ll} 0, & \textrm{if $X$ is of unbounded
variation,} \\ c^{-1}, & \textrm{if $X$ is of bounded variation,}
\end{array} \right.  
\end{split}
\end{align*}
and a similar result holds for $\mathbb{W}^{(q)}$.
\end{enumerate}
\end{remark}

Using the results in
\cite{7.}, the expected NPV \eqref{v_pi} can be written as below.

\begin{lemma} \label{lemma_NPV} 
For $q > 0$, $b \geq 0$ and $x \in \R$, we have
\begin{align*}
v_b(x)
=&-\delta\overline{\mathbb{W}}^{(q)}(x-b) + \beta \Big(\overline{Z}^{(q)}(x) + \frac {\psi'(0+)} q \Big)\\
&+ \beta \delta\int_b^x\mathbb{W}^{(q)}(x-y)Z^{(q)}(y)\diff y \notag\\
&-
\frac{f(b)}{ q}
\Big(Z^{(q)}(x)+q\delta\int_b^x\mathbb{W}^{(q)}(x-y)W^{(q)}(y)\diff y\Big),
\end{align*}
\begin{align*}
\hspace{-28mm}\text{where}\ \ f(b) :=   \frac{  \beta Z^{(q)}(b) -1  + \beta q \int_0^{\infty}e^{-\varphi(q)y}W^{(q)}(y+b)\diff y } { \varphi(q)  \int_0^{\infty}e^{-\varphi(q)y}W^{(q)}(y+b)\diff y}.
\end{align*}
\end{lemma}
\begin{proof}
The result follows by Corollaries 4.4 and 5.5 of \cite{7.} and the fact that
\begin{align*}
\int_0^{\infty}e^{-\varphi(q)y}Z^{(q)}(y+b)\diff y = \frac {Z^{(q)}(b)} {\varphi(q)} + \frac q {\varphi(q)} \int_0^{\infty}e^{-\varphi(q)y}W^{(q)}(y+b)\diff y.
\end{align*}
\hfill\qed

\end{proof}

\begin{remark}\label{remark_when_b_zero} As \eqref{scale_function_laplace} gives $\int_0^{\infty}e^{-\varphi(q)y}W^{(q)}(y)\diff y = (\delta \varphi(q))^{-1}$ for the case when $b=0$, we have
\begin{align}
f(0) =   \frac{  \beta -1  + \beta q \int_0^{\infty}e^{-\varphi(q)y}W^{(q)}(y)\diff y } { \varphi(q)  \int_0^{\infty}e^{-\varphi(q)y}W^{(q)}(y)\diff y} =   \delta \Big(  \beta -1  +  \frac {\beta q} {\delta \varphi(q)} \Big). \label{f_0}
\end{align}
By this, \eqref{RLqp} and \eqref{RLqp2}, we get
\begin{align*}
v_0(x)
=&-\delta\overline{\mathbb{W}}^{(q)}(x) + \beta \Big(\overline{Z}^{(q)}(x) + \frac {\psi'(0+)} q \Big)  + \beta \Big(\overline{\mathbb{Z}}^{(q)}(x)-\overline{Z}^{(q)}(x)+\delta\overline{\mathbb{W}}^{(q)}(x) \Big)\notag\\
&+\left( 1 -\beta    - \frac {\beta q} {\delta \varphi(q)}  \right)\frac{Z^{(q)}(x)+q \Big( \overline{\mathbb{W}}^{(q)}(x)-\overline{W}^{(q)}(x) \Big)} { q } \delta \\
=& \frac \delta q  + \beta \Big( \frac {\psi'(0+)} q  + \overline{\mathbb{Z}}^{(q)}(x) - \frac \delta q -     \frac {1} {\varphi(q)}   \mathbb{Z}^{(q)}(x) \Big).
\end{align*}

\end{remark}

\section{Selection of the Candidate Threshold} \label{section_threshold}
We shall choose the candidate threshold $b^*$ so that the corresponding expected NPV $v_{b^*}$ can be smooth at $b^*$. Lemma \ref{lemma_NPV} and
integration by parts imply that
\begin{align*}
&v_b(x)\\
=&-\delta\overline{\mathbb{W}}^{(q)}(x-b) + \beta \Big(\overline{Z}^{(q)}(x) + \frac {\psi'(0+)} q \Big)\\
  &+ \beta \delta\Big( Z^{(q)} (b) \overline{\mathbb{W}}^{(q)} (x-b) + q \int_b^x \overline{\mathbb{W}}^{(q)} (x-y) W^{(q)} (y) \diff y \Big)
 \notag\\
&-
\frac {f(b)} q
\Big[Z^{(q)}(x)+q\delta \Big( W^{(q)} (b) \overline{\mathbb{W}}^{(q)} (x-b) +  \int_b^x \overline{\mathbb{W}}^{(q)} (x-y) W^{(q) \prime} (y) \diff y \Big)\Big].
\end{align*}
By differentiating this, 
we have
\begin{align} \label{v_b_prime_kazu}
v_b'(x)
=&-\delta \mathbb{W}^{(q)}(x-b) + \beta Z^{(q)}(x)\nonumber \\
&+ \beta \delta  \Big[  Z^{(q)} (b) \mathbb{W}^{(q)} (x-b) + q \int_b^x \mathbb{W}^{(q)} (x-y) W^{(q)} (y) \diff y \Big]\nonumber \\
& {-f(b)}
\Big[ W^{(q)}(x)+  \delta \Big( W^{(q)} (b) \mathbb{W}^{(q)} (x-b) +  \int_b^x \mathbb{W}^{(q)} (x-y) W^{(q) \prime} (y) \diff y \Big) \Big],
\end{align}
which is continuous for $x \neq 0, b$.

In particular, for $x < b$, by \eqref{z_below_zero}, we get
\begin{align} \label{v_prime_below_b}
v_b'(x)
&= \beta Z^{(q)}(x)  -  W^{(q)}(x) f(b).
\end{align}
It follows that
\begin{align} \label{v_prime_diff}
v_b'(b+) - v_b'(b-)
&= \delta \mathbb{W}^{(q)}(0) g(b),
\end{align}
where
\begin{align} \label{g_f_connection}
\begin{split}
g(b) &:=  \beta   Z^{(q)} (b) - 1  - W^{(q)}(b) f(b)
\\
&=\left( \beta Z^{(q)}(b) -1   \right) \Big( 1- \frac{W^{(q)}(b) } { \varphi(q)  \int_0^{\infty}e^{-\varphi(q)y}W^{(q)}(y+b)\diff y} \Big) - \frac {\beta q W^{(q)}(b)} {\varphi(q)}.
\end{split}
\end{align}

For the case of bounded variation, where $\mathbb{W}^{(q)}(0) > 0$ (see Remark \ref{remark_scale_function_properties}(ii)), it is straightforward to see that $v_b$ is continuously differentiable at $b$ if and only if $g(b) = 0$.

%

For the case of unbounded variation, where $\mathbb{W}^{(q)}(0) = 0$,
by differentiating \eqref{v_b_prime_kazu}, we have, for $x \neq b$, that
\begin{align*}
v^{b \prime \prime}(x)
=& -\delta \mathbb{W}^{(q) \prime}(x-b) + \beta q W^{(q)}(x)\\
 &+ \beta \delta  \Big[  Z^{(q)} (b) \mathbb{W}^{(q)\prime} (x-b) + q \int_b^x \mathbb{W}^{(q)\prime} (x-y) W^{(q)} (y) \diff y \Big] \\
& {-f(b)}
\Big[ W^{(q)\prime}(x)+  \delta \Big( W^{(q)} (b) \mathbb{W}^{(q)\prime} (x-b)\\
&+ \int_b^x \mathbb{W}^{(q)\prime} (x-y) W^{(q) \prime} (y) \diff y \Big) \Big],
\end{align*}
which is continuous for $x \neq 0, b$ and hence $v^{b \prime \prime}(b+) - v^{b \prime \prime}(b-)
= \delta \mathbb{W}^{(q) \prime}(0+) g(b)$.

These observations, together with the smoothness of the scale function on $\R \backslash \{0\}$ as in Remark  \ref{remark_scale_function_properties}, are summarized as follows.
\begin{lemma} \label{lemma_smooth_fit} Suppose that there exists $b > 0$ such that $g(b) = 0$. Then, $v_b$ is continuously (resp.\ twice continuously) differentiable on $]0, \infty[$ when $X$ is of bounded (resp.\ unbounded) variation.
\end{lemma}

\begin{remark}[slope at $b$] \label{remark_slope_at_b}
The condition $g(b) = 0$ is equivalent to $v_b'(b-) = 1$, if $b > 0$.  Indeed, by \eqref{v_prime_below_b}, \eqref{v_prime_diff}, and \eqref{g_f_connection}, we have
\begin{align*}
v_b'(b-) &= 1+g(b), \quad b > 0, \\
v_b'(b+) &=  1 + (1+\delta \mathbb{W}^{(q)}(0)) g(b),
\quad b \geq 0.
\end{align*}

\end{remark}

\begin{remark}[Continuity/smoothness at zero] \label{remark_smoothness_zero}
(i) 
By Lemma \ref{lemma_NPV}, we have that $v_{b}$ is continuous at zero for $b \geq 0$.

(ii) If $b > 0$, \eqref{v_prime_below_b} gives $v_{b}'(0+)
= \beta  -  W^{(q)}(0) f(b) = \beta = v_{b}'(0-)$ for the case of unbounded variation.

%
%
%
\end{remark}

Let us define our candidate threshold by
\begin{equation}\label{defbthreshold}
b^* := \inf \{ b \geq 0:
g(b) \leq 0 \},
\end{equation}
with the convention that $\inf \varnothing = \infty$.

\begin{lemma} \label{lemma_criteria_zero}
We have $b^* = 0$ if and only if $X$ is of bounded variation and
\begin{align}
\beta  -1  +\left( 1 -\beta    -  \frac {\beta q} {\delta \varphi(q)} \right)  \frac \delta c  \leq 0. \label{criteria_b_0}
\end{align}
\end{lemma}
\begin{proof}
By the definition of $b^*$ as in \eqref{defbthreshold}, we have that $b^* = 0$ if and only if $g(0) \leq 0$, where $g(0)= \beta  -1  +\left( 1 -\beta    -  \frac {\beta q} {\delta \varphi(q)} \right)  \delta W^{(q)}(0)$ by \eqref{f_0} and \eqref{g_f_connection}.

For the case of unbounded variation (where $W^{(q)}(0) = 0$), $g(0) = \beta - 1 > 0$ and hence $b^* > 0$. On the other hand, for the case of bounded variation, by Remark  \ref{remark_scale_function_properties}(ii), $b^* = 0$ if and only if \eqref{criteria_b_0} holds.\hfill\qed
\end{proof}

To continue, we can further show that $b^* < \infty$.
\begin{lemma} \label{g_probabilisitic}
(i) Define $h(b) :=  1- \frac{W^{(q)}(b) } { \varphi(q)  \int_0^{\infty}e^{-\varphi(q)y}W^{(q)}(y+b)\diff y}, \quad b \geq 0$. Then, for $b \geq 0$,
\begin{align} \label{eqn_frac_g_h}
\frac {g(b)} {h(b)} 
 &=  \beta  Z^{(q)} (b)-1  -  {\beta q}   {W^{(q)}(b)}
  \frac {   \int_0^{\infty}e^{-\varphi(q)y}W^{(q)}(y+b)\diff y}  {  \int_0^{\infty}e^{-\varphi(q)y}W^{(q)\prime}(y+b)  \diff y  }\nonumber\\
  &= \beta\E_b \left[e^{-q \kappa_0^{b,-}}1_{\{\kappa_0^{b,-} <\infty\}}\right] - 1,
\end{align}
where $\kappa_0^{b,-}:=\inf\{t>0:U_t^{b}<0\}$ and $U^{b}$ is the refracted \lev process of {\normalfont\cite{12.}}, which is the unique strong solution to the stochastic differential equation
$U_t^{b}=X_t-\delta\int_0^t1_{\{U_s^{b}>b\}} \diff s$, for  $t \geq 0$.
\\
(ii)
We have $0 \leq b^* < \infty$.
\end{lemma}
\begin{proof}
(i) We have, by \eqref{g_f_connection}, that
\begin{align} \label{g_transform}
\begin{split}
g(b)
&= (\beta  Z^{(q)} (b)-1 )  h(b) - \frac {\beta q} {\varphi(q)} W^{(q)}(b).
\end{split}
\end{align}
On the other hand, by integration by parts,
\begin{align} \label{h_inverse}
\begin{split}
h(b)^{-1}
&=  \varphi(q)
  \frac {   \int_0^{\infty}e^{-\varphi(q)y}W^{(q)}(y+b)\diff y}  {  \int_0^{\infty}e^{-\varphi(q)y}W^{(q)\prime}(y+b)  \diff y  }.
\end{split}
\end{align}
Therefore, due to \eqref{g_transform} and \eqref{h_inverse}, we obtain the first equality of \eqref{eqn_frac_g_h}.
The second equality of \eqref{eqn_frac_g_h} holds by Theorem 5 (ii) of \cite{12.}.

(ii) For $b \geq 0$, because $W^{(q)}$ is strictly increasing on $[0, \infty[$, we get
\begin{align*}
h(b) =  1- \frac{W^{(q)}(b) } { \varphi(q)  \int_0^{\infty}e^{-\varphi(q)y}W^{(q)}(y+b)\diff y} > 1- \frac{W^{(q)}(b) } { \varphi(q)  \int_0^{\infty}e^{-\varphi(q)y}W^{(q)}(b)\diff y} = 0.
\end{align*}
Using the fact that $U^b\geq Y$ and hence that $\kappa_0^{b,-}$ is dominated from below by the down-crossing time of $Y$, we have the convergence $\E_b [e^{-q \kappa_0^{b,-}}1_{\{\kappa_0^{b,-} <\infty\}}] \rightarrow 0$ as $b \rightarrow \infty$. This and \eqref{eqn_frac_g_h} imply that $\lim_{b \rightarrow \infty}g(b) / h(b) = -1$. Hence, by the positivity of $h$, $g(b)$ must be negative for a sufficiently large $b$. Consequently, it follows that $0 \leq b^* < \infty$.\hfill\qed
\end{proof}


\section{Verification of Optimality} \label{section_verification}
In this section, we provide a rigorous verification argument for the choice of $b^{\ast}$ defined in $(\ref{defbthreshold})$ such that the value function of the stochastic control problem $(\ref{control:value})$ can be achieved.

With the selected barrier $b^*$, by Lemma \ref{lemma_NPV}, our value function becomes
\begin{align}\label{vf_new_label}
\begin{split}
v_{b^*}(x)
=&-\delta\overline{\mathbb{W}}^{(q)}(x-b^*) + \beta \Big(\overline{Z}^{(q)}(x) + \frac {\psi'(0+)} q \Big)  \\ &+ \beta \delta\int_{b^*}^x\mathbb{W}^{(q)}(x-y)Z^{(q)}(y)\diff y \\
&- \frac{f(b^*)}{q} \Big(Z^{(q)}(x)+q\delta\int_{b^*}^x\mathbb{W}^{(q)}(x-y)W^{(q)}(y)\diff y \Big).
\end{split}
\end{align}
Here, for the case $b^* > 0$,  because $g(b^*) = 0$, by \eqref{g_f_connection} and \eqref{eqn_frac_g_h}, we derive that
\begin{align} \label{b_star_relation}
f(b^*)
 = \frac {\beta  Z^{(q)} (b^*) -1}{W^{(q)}(b^*)} = \beta q
  \frac {   \int_0^{\infty}e^{-\varphi(q)y}W^{(q)}(y+b^*)\diff y}  {  \int_0^{\infty}e^{-\varphi(q)y}W^{(q)\prime}(y+b^*)  \diff y  }.
\end{align}
For the case $b^* = 0$, $v_{b^*} = v_0$ is already given in Remark \ref{remark_when_b_zero}.

Our goal is to prove the main result of this paper given below.
\begin{theorem} \label{main_theorem}The strategy $\pi^{b^*}$ is optimal and the value function of the stochastic control problem $(\ref{control:value})$ is given by $v = v_{b^*}$.
\end{theorem}

Let $\mathcal{L}$ be the infinitesimal generator associated with
the process $X$ applied to a $C^1$ (resp.\ $C^2$) function $F$ for the case where $X$ is of bounded (resp.\ unbounded) variation, i.e., for $x \in \R$,
\begin{align*} 
\begin{split}
\mathcal{L} F(x) &:= \gamma F'(x) + \frac 1 2 \sigma^2 F''(x) \\ &+ \int_{(-\infty,0)} \left[ F(x+z) - F(x) -  F'(x) z 1_{\{-1 < z < 0\}} \right] \Pi(\diff z). 
\end{split}
\end{align*}
Further, let $\mathcal{L}_Y$ be that of $Y_t := X_t - \delta t$. We have $\mathcal{L}_Y F(x) = \mathcal{L} F(x) - \delta F'(x)$.

To show the optimality, it suffices to verify variational inequalities. The proof of the next lemma is omitted as it is essentially the same as the spectrally positive case in Lemma 4.2 of \cite{5.}.
Here we slightly relax the assumption on the smoothness at zero, which can be achieved by applying the Meyer-It\^o formula as in Theorem 4.71 of \cite{15.}. We refer to \cite{1.,6.,16.,17.} for other stochastic control problems and verification lemmas with spectrally one-sided \lev processes.

\begin{lemma}[Verification lemma]
\label{verificationlemma}
Suppose $\hat{\pi} \in \mathcal{A}$ 
such that $v_{\hat{\pi}}$ is sufficiently smooth on $]0,\infty[$, continuous on $\R$, and, for the case of unbounded variation, continuously differentiable at zero. In addition, we assume that 
\begin{align}
\label{HJB-inequality}
\sup_{0\leq r\leq\delta} \big((\mathcal{L} - q)v_{\hat{\pi}}(x)-rv'_{\hat{\pi}}(x)+r \big) &\leq 0, \quad x > 0, \notag\\
v'_{\hat{\pi}}(x)&\leq\beta, \quad x > 0, \\
\inf_{x\geq 0}v_{\hat{\pi}}(x) &> -m,\qquad\text{for some $m>0$.}\notag
\end{align} 
Then, $v_{\hat{\pi}}(x)=v(x)$ for all $x\geq0$, and hence, $\hat{\pi}$ is an optimal strategy.
\end{lemma}


We shall first compute the generator parts.

\begin{lemma} \label{generator_lemma}
Fix $b \geq 0$.
(i) If $b > 0$, we have $(\mathcal{L}-q) v_b(x) = 0$ for $0 < x < b$. \\
(ii) We have $(\mathcal{L}_Y-q) v_b(x) + \delta=(\mathcal{L}-q) v_b(x) + \delta (1- v_b'(x)) = 0$ for $x > b$.
\end{lemma}
\begin{proof}
(i) For $0 < x < b$,
Theorem 2.1 in \cite{8.}
leads to
\begin{align*}
(\mathcal{L}-q)v_{b}(x)
&=\beta (\mathcal{L}-q) \Big(\overline{Z}^{(q)}(x) + \frac {\psi'(0+)} q \Big) -\frac{f(b)}{q}
(\mathcal{L}-q) Z^{(q)}(x) =0.\end{align*}

(ii)
On the other hand, for $x > b$, Theorem 2.1 in \cite{8.}
implies that
\begin{align*}
&(\mathcal{L}_Y-q)  \overline{\mathbb{W}}^{(q)}(x-b) = q^{-1}(\mathcal{L}_Y-q)  (\mathbb{Z}^{(q)}(x-b) -1) =1, \\
&(\mathcal{L}_Y-q) \Big(\overline{Z}^{(q)}(x) + \frac {\psi'(0+)} q \Big) = - \delta \frac \partial {\partial x}\Big(\overline{Z}^{(q)}(x) + \frac {\psi'(0+)} q \Big)  = - \delta Z^{(q)}(x), \\
&(\mathcal{L}_Y-q) Z^{(q)}(x) =  - \delta Z^{(q)\prime}(x)  = - \delta q W^{(q)} (x).
\end{align*}
In addition, we get $(\mathcal{L}_Y-q)  \Big( \int_b^x\mathbb{W}^{(q)}(x-y)l(y)\diff y \Big) = l(x)$ by the argument in the proof of Lemma 4.5 of \cite{18.}, for $l = Z^{(q)}, W^{(q)}$. Applying these in \eqref{vf_new_label}, we have that claim (2) holds. \hfill\qed
\end{proof}


\begin{lemma} \label{lemma_slope}
For the threshold $b^*$ defined by $(\ref{defbthreshold})$, we have $\beta\geq v_{b^*}'(x) \geq 1$ for $x < b^*$, and $0 \leq v_{b^*}'(x) \leq 1$ for $x \geq b^*$.
\end{lemma}
\begin{proof}
Step (i): Suppose $b^* > 0$.
By \eqref{v_b_prime_kazu} and \eqref{b_star_relation}, we have
\begin{align}
v_{b^*}'(x)
=&-\delta \mathbb{W}^{(q)}(x-b^*) + \beta Z^{(q)}(x) \notag\\
 &+ \beta \delta  \Big[  Z^{(q)} (b^*) \mathbb{W}^{(q)} (x-b^*) + q \int_{b^*}^x \mathbb{W}^{(q)} (x-y) W^{(q)} (y) \diff y \Big] \nonumber \\
&-  \frac {\beta  Z^{(q)} (b^*)-1}{W^{(q)}(b^*)}
\Bigg[ W^{(q)}(x)+  \delta \Big( W^{(q)} (b^*) \mathbb{W}^{(q)} (x-b^*) \notag \\
& +  \int_{b^*}^x \mathbb{W}^{(q)} (x-y) W^{(q) \prime} (y) \diff y \Big) \Bigg] \nonumber \\
%
%
\label{v_der_2}
=&\beta Z^{(q)}(x) + \beta\delta q\int_{b^*}^x\mathbb{W}^{(q)}(x-y)W^{(q)}(y)\diff y  \notag\\
&-\beta q\frac{\int_0^{\infty}e^{-\varphi(q)y}W^{(q)}(y+b^*)\diff y}{\int_0^{\infty}e^{-\varphi(q)y}W^{(q)\prime}(y+b^*)\diff y}\bigg(W^{(q)}(x)\notag\\
&+\delta\int_{b^*}^x\mathbb{W}^{(q)}(x-y)W^{(q)\prime}(y)\diff y \bigg)\notag\\
=&\beta\E_x\left[e^{-q \kappa_0^{b^*,-}}1_{\{\kappa_0^{b^*,-}<\infty\}}\right],
\end{align}
where the second equality holds by the second equality of \eqref{b_star_relation},
 and the last equality holds by Theorem 5 (ii) in \cite{12.}.
Thanks to \eqref{v_der_2}, we deduce that $0 \leq v_{b^*}'(x) \leq \beta = v_{b^*}'(0-)$ and $v_{b^*}'(x)$ is non-increasing for $x>0$. This and $v_{b^*}'(b^*) = 1$ implied by Remark \ref{remark_slope_at_b} complete the proof.

Step (ii): Suppose $b^*=0$ (then, necessarily $X$ is of bounded variation by Lemma \ref{lemma_criteria_zero}). By Remark \ref{remark_when_b_zero}, we have, for $x \neq 0$,
\begin{align*}
v_0'(x)
=&  \beta \Big( \mathbb{Z}^{(q)}(x) -     \frac {q} {\varphi(q)}    \mathbb{W}^{(q)}(x) \Big), \\
v_0^{\prime \prime}(x+)
=& \beta q \mathbb{W}^{(q)}(x)  \Big(  1 -     \frac {1} {\varphi(q)}    \frac {\mathbb{W}^{(q)\prime}(x+)} {\mathbb{W}^{(q)}(x)}  \Big).
\end{align*}
It is known that $x \mapsto \mathbb{W}^{(q)\prime}(x+) /\mathbb{W}^{(q)}(x) $ is monotonically decreasing in $x$ as in (8.18) and Lemma 8.2 of \cite{19.}, and it converges to $\varphi(q)$.  Hence, $v_0''(x+) < 0$, which implies that $v_0$ is concave.

On the other hand, we have $v_0^{\prime}(0+) =  1 + (1+\delta \mathbb{W}^{(q)}(0)) g(0)$ by Remark \ref{remark_slope_at_b}. As $g(0) \leq 0$ (see Lemma \ref{lemma_criteria_zero}), we have $v_0^{\prime}(0+) \leq 1$. It follows that $v_0'(x)  \leq 1$ for all $x$.
Finally, we have $v_0'(x) \xrightarrow{x \uparrow \infty} 0$ because $ \mathbb{Z}^{(q)}(x) -  {q}  \mathbb{W}^{(q)}(x) /  {\varphi(q)}  $ vanishes in the limit by Theorem 8.1 (ii) of \cite{19.}. Hence, we have $v'_0(x) \geq 0$.\hfill\qed

\end{proof}

\begin{remark} \label{lemma_bound}
By the monotonicity of $v_{b^*}$ in view of Lemma \ref{lemma_slope} and Assumption \ref{assump_mean_finite},  we have $\inf_{x\geq0}v_{b^*}(x) \geq  v_{b^*}(0) > -\infty$.
\end{remark}

\begin{proof}[of Theorem {\normalfont\ref{main_theorem}}]
We shall show that $v_{b^*}$ satisfies all conditions given in Lemma \ref{verificationlemma}. First, by Lemma \ref{lemma_smooth_fit} and Remark \ref{remark_smoothness_zero}, the desired continuity/smoothness of $v_{b^*}$ holds.

It is left to verify the variational inequalities \eqref{HJB-inequality}. Lemma \ref{lemma_slope} leads to
\begin{align*}
\sup_{0\leq r\leq\delta} r \big(1-v'_{b^*}(x) \big) = \left\{ \begin{array}{ll} \delta \big(1-v'_{b^*}(x) \big) \leq \delta, & \textrm{if }  x > b^*, \\ 0, & \textrm{if }   0 < x \leq b^*. \end{array} \right.
\end{align*}
This and Lemma \ref{generator_lemma} yield the first item of \eqref{HJB-inequality} with equality.
The second item holds by Lemma \ref{lemma_slope}.
Lastly, the third item holds by Remark \ref{lemma_bound}. \hfill\qed
\end{proof}

\begin{remark}
Regardless of the negative jumps of $X$, our conclusion interestingly indicates that our conjectured threshold strategy is still the optimal strategy. However, as the term $\psi'(0+) = \E X_1 = \gamma + \int_{]-\infty,-1]}z \Pi (\diff z)$ appears in the value function, the negative jumps clearly have direct impacts on the optimal solution.

Another important impact of the jumps can be seen in the bounded variation case, where the optimal threshold can be $b^*=0$, which implies that it is optimal to always pay dividends. This outcome does not occur in the classical Brownian motion model.\end{remark}

\section{Numerical Examples} \label{section_numerics}
We conclude this paper with a sequence of numerical experiments on the underlying process modeled by the spectrally negative \lev process with phase-type jumps of the form that
$ X_t - X_0= c t+\sigma B_t - \sum_{n=1}^{N_t} Z_n$, for $0\le t <\infty$. 
Here, $B=( B_t; t\ge 0)$ is a standard Brownian motion, $N=(N_t; t\ge 0 )$ is a Poisson process with arrival rate $\kappa$, and  $Z = ( Z_n; n = 1,2,\ldots )$ is an i.i.d.\ sequence of phase-type random variables that approximate the Weibull distribution with shape parameter $2$ and scale parameter $1$ (see \cite{16.} for the parameters of the phase-type distribution and also \cite{20.} for the accuracy of approximation). The processes $B$, $N$, and $Z$ are assumed to be mutually independent.  We refer the reader to \cite{14.} and \cite{20.} for the forms of the corresponding scale functions. We consider \textbf{Case} 1 (unbounded variation) with $\sigma = 0.2$ and $c = 2$ and \textbf{Case} 2  (bounded variation) with $\sigma = 0$ and $c = 4$.  For other parameters, let us set $\kappa = 2$, $q=0.05$, $\beta = 1.5$ and $\delta = 1$ unless stated otherwise.

Recall that the optimal threshold $b^*$ is given by \eqref{defbthreshold}. In Figure \ref{plot_g_h}, we plot the function $b \mapsto g(b)/h(b)$ (recall that $h$ is uniformly positive) for various values of $\beta$ for \textbf{Cases} 1 and 2.  For the case $g(0) \leq 0$ (and hence $g(0)/h(0) \leq 0$), we have $b^* = 0$. Otherwise, $g/h$ is monotonically decreasing and $b^*$ becomes the value at which $g$ (and $g/h$) vanish. As observed in Lemma \ref{lemma_criteria_zero}, for \textbf{Case} 1 (unbounded variation), $b^* > 0$ for any value of $\beta > 1$ while in \textbf{Case} 2 (bounded variation case), $b^* = 0$ if $\beta$ is sufficiently close to $1$.  In order to confirm the optimality of the selected threshold strategy $\pi^{b^*}$,
we plot, as shown in Figure \ref{plot_optimality} (for $\beta = 1.5$), the value function $v_{b^*}$  together with $v_b$ for $b \neq b^*$.  For \textbf{Case} 1, we have $b^* > 0$; while for \textbf{Case} 2, we have $b^* = 0$. It is illustrated in the figure that $v_{b^*}$ satisfactorily dominates $v_b$ uniformly in $x$.

In Figure \ref{plot_sensitivity}, we present the sensitivity of the optimal solutions with respect to parameters $\beta$ and $\delta$ focusing on \textbf{Case} 1.  On the left panel, we plot $v_{b^*}$ for $\beta$ ranging from $1.01$ to $3$. The graph indicates that the value function decreases in $\beta$ uniformly in $x$ and that the optimal threshold $b^*$ increases as $\beta$ increases.  On the right panel, we show $v_{b^*}$ for $\delta$ varying from $0.01$ to $3$ along with results in the case without the absolutely continuous assumption as in \cite{1.}. It is observed that the value function converges increasingly to that in \cite{1.}.  The convergence of $b^*$ to the optimal barrier in \cite{1.} is also confirmed.

\section{Conclusions}\label{section_conc}
We solved the dividend problem with capital injection under the constraint that the cumulative dividend strategy is absolutely continuous. In particular, we proved that the solution is a refraction-reflection strategy that reflects the surplus from below at zero and decreases the drift at a suitable threshold.

It is noted that the methods and results in this current paper can potentially be applied in other related stochastic control problems driven by one dimensional spectrally one-sided \lev processes.  In inventory/cash management control problems as in \cite{21.}, it is of interest to pursue the optimality of refraction-reflection strategies under suitable absolutely continuous assumptions.  Using the results in \cite{7.}, smooth fit and verification are expected to be carried out in an efficient way as  in this current paper.

\begin{figure}[htbp]
\vspace*{0mm}
\hspace*{-10mm}
\begin{center}
\begin{minipage}{1.0\textwidth}
\centering
\begin{tabular}{cc}
 \includegraphics[scale=0.3]{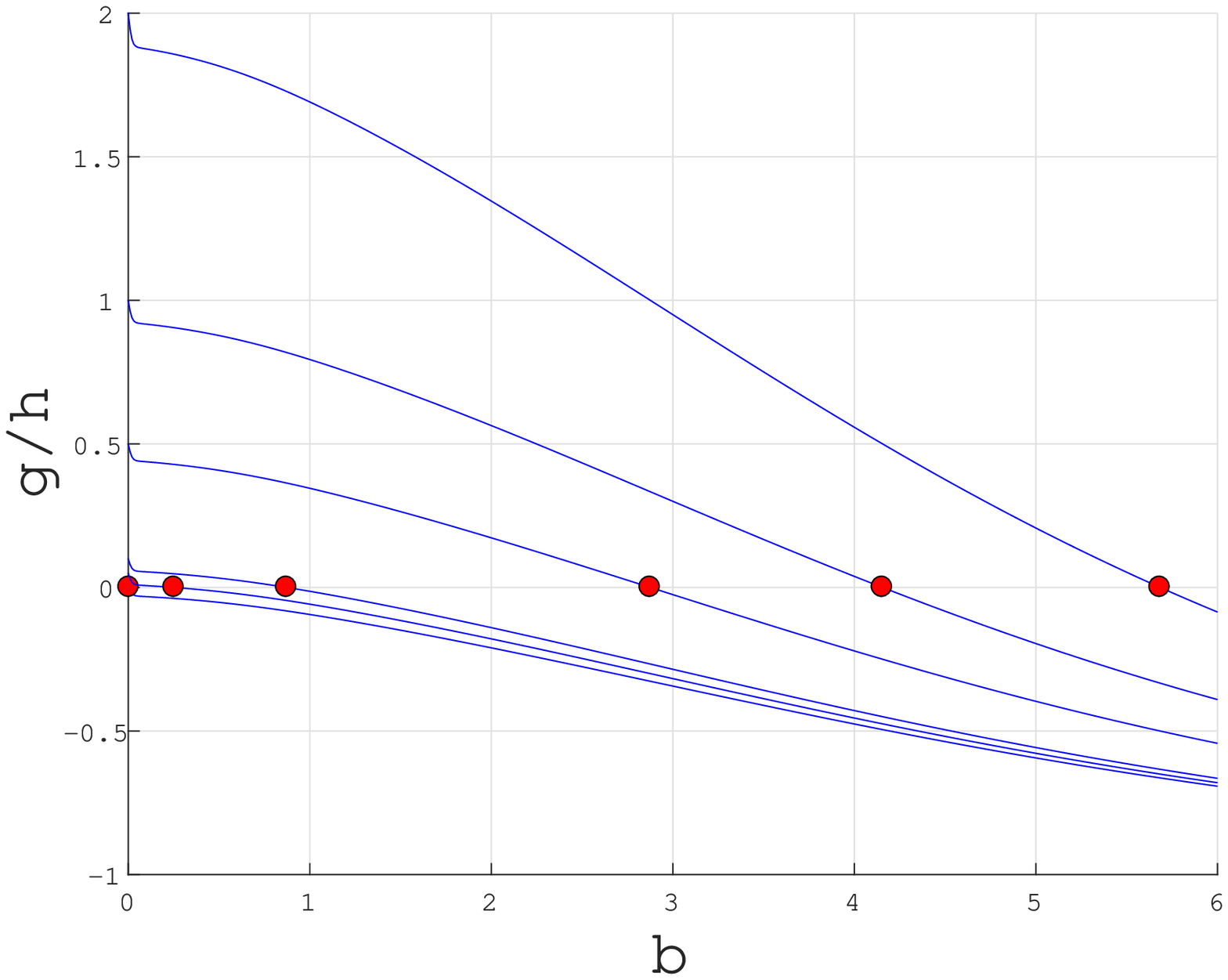} & \includegraphics[scale=0.3]{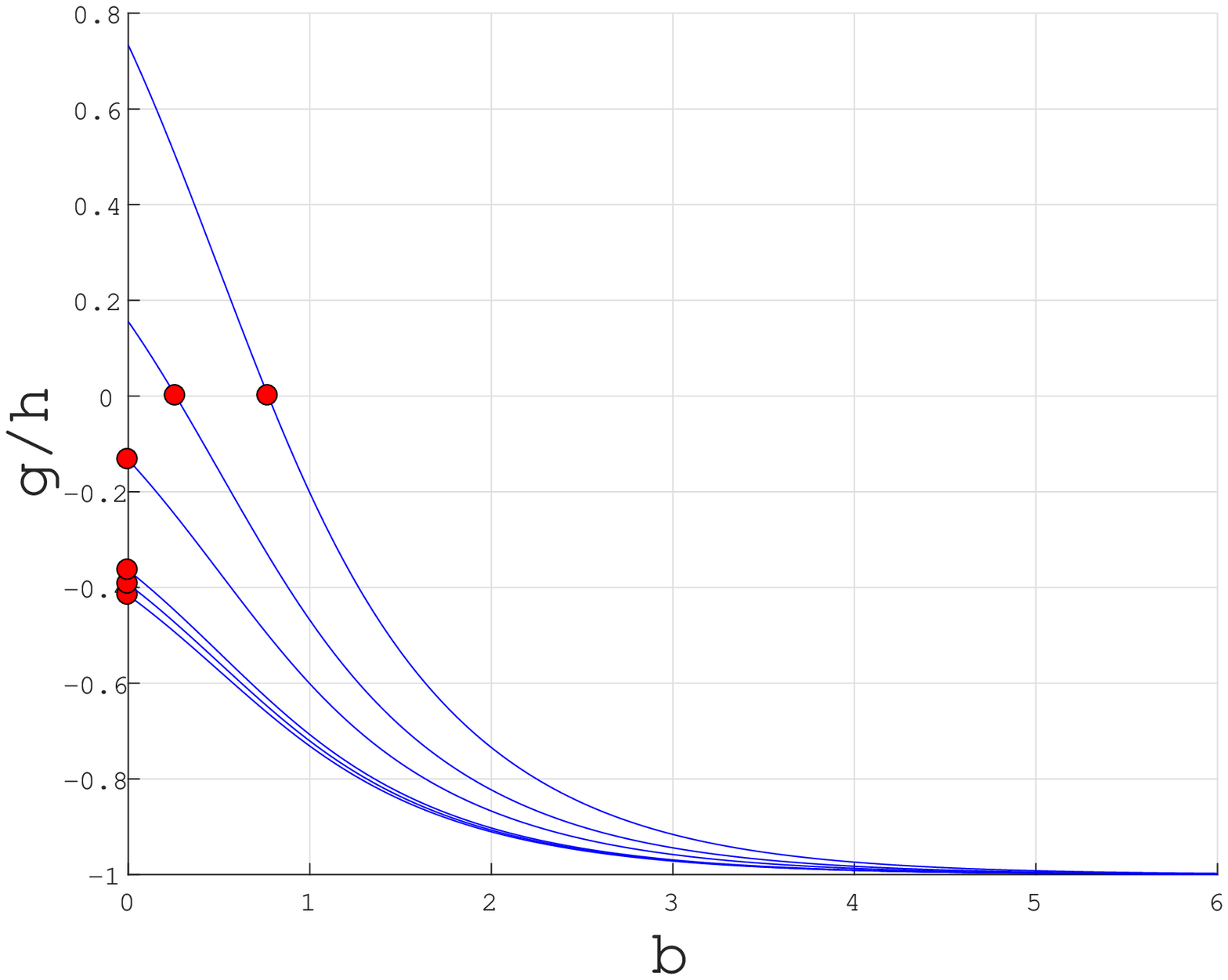}    \\
 \textbf{Case} 1 & \textbf{Case} 2 \end{tabular}
\end{minipage}
\caption{Plots of $b \mapsto g(b)/h(b)$ for \textbf{Case} 1 (left) and \textbf{Case} 2 (right) for $\beta = 1.01, 1.05,1.1,1.5,2$ and $3$, which correspond to curves from bottom to top. The circles indicate the points at $b^*$.
} \label{plot_g_h}
\end{center}
\end{figure}

\begin{figure}[htbp]
\vspace*{0mm}
\hspace*{-10mm}
\begin{center}
\begin{minipage}{1.0\textwidth}
\centering
\begin{tabular}{cc}
 \includegraphics[scale=0.30]{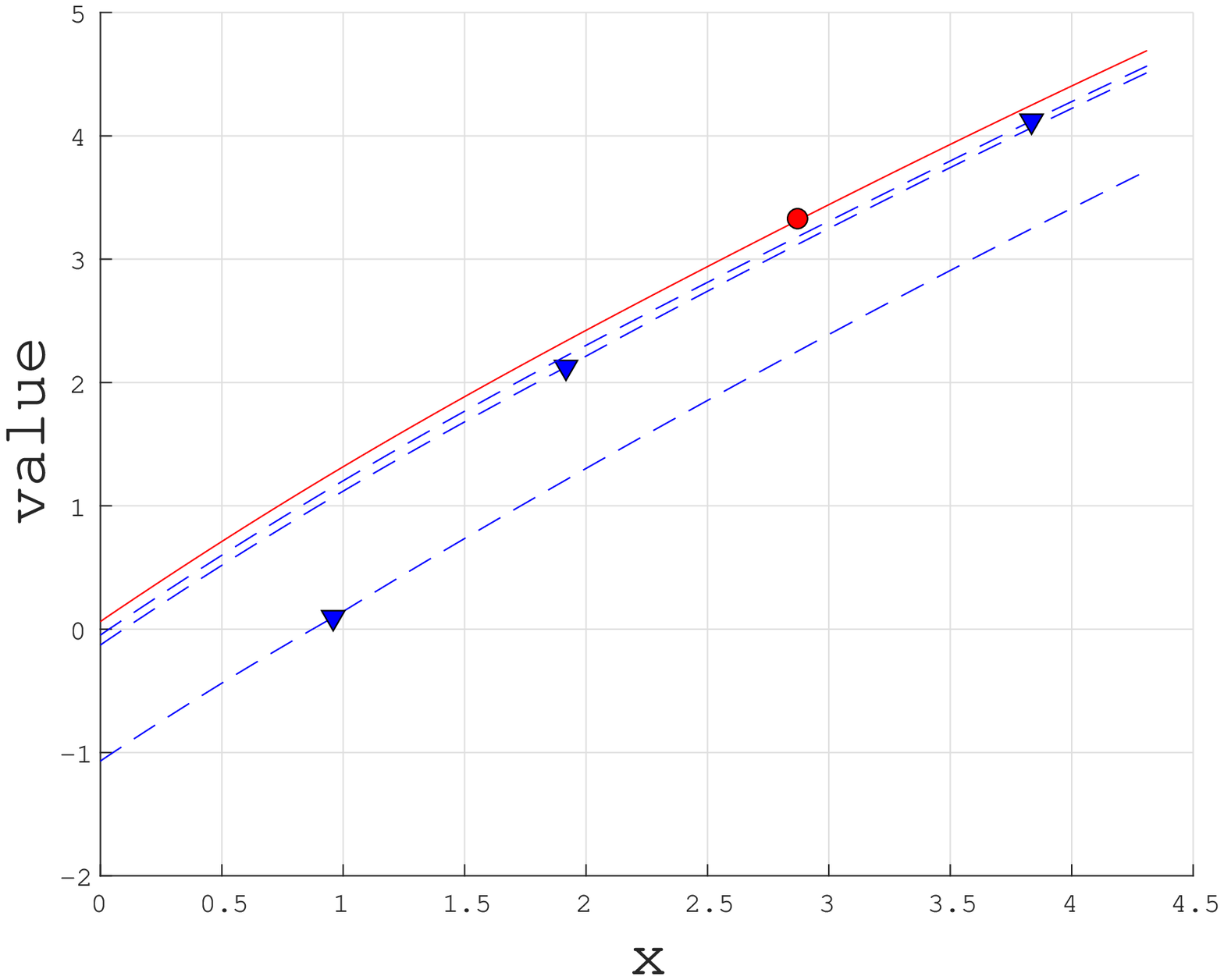} & \includegraphics[scale=0.30]{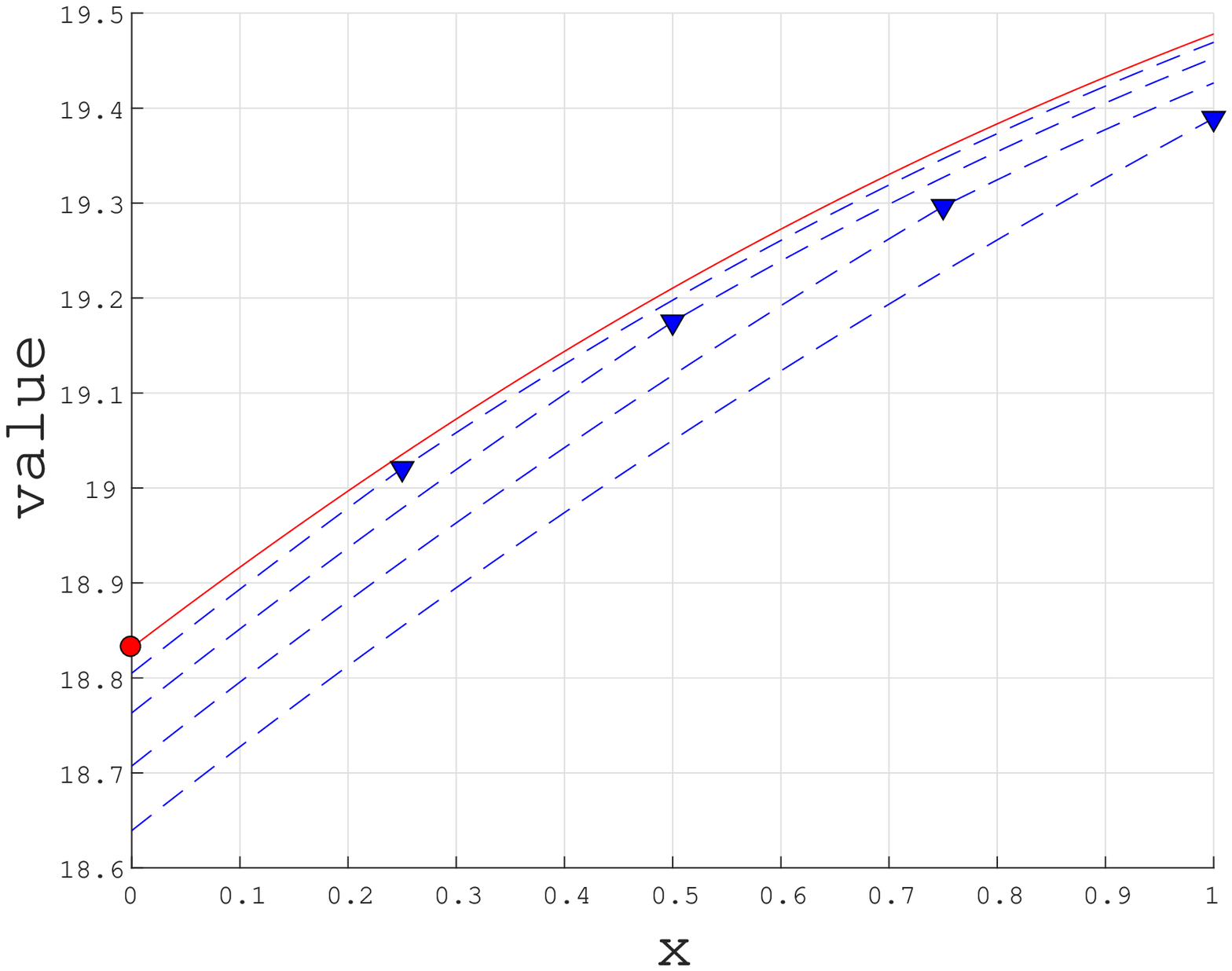}    \\
 \textbf{Case} 1 & \textbf{Case} 2 \end{tabular}
\end{minipage}
\caption{(Left) plots of $x \mapsto v_{b^*}(x)$ (solid) for \textbf{Case} 1 along with $v_b$ (dotted) for $b = b^*/3, 2b^*/3$ and $4b^*/3$, which correspond to dotted curves from bottom to top. (Right) plots of $x \mapsto v_{b^*}(x)$ (solid) for \textbf{Case} 2 along with $v_b$ (dotted) for $b = 1/4, 1/2, 3/4$ and $1$, which correspond to dotted curves from top to bottom. The circles and down-pointing triangles indicate the points at $b^*$ and $b$, respectively. } \label{plot_optimality}
\end{center}
\end{figure}

\begin{figure}[htbp]
\vspace*{0mm}
\hspace*{-10mm}
\begin{center}
\begin{minipage}{1.0\textwidth}
\centering
\begin{tabular}{cc}
 \includegraphics[scale=0.30]{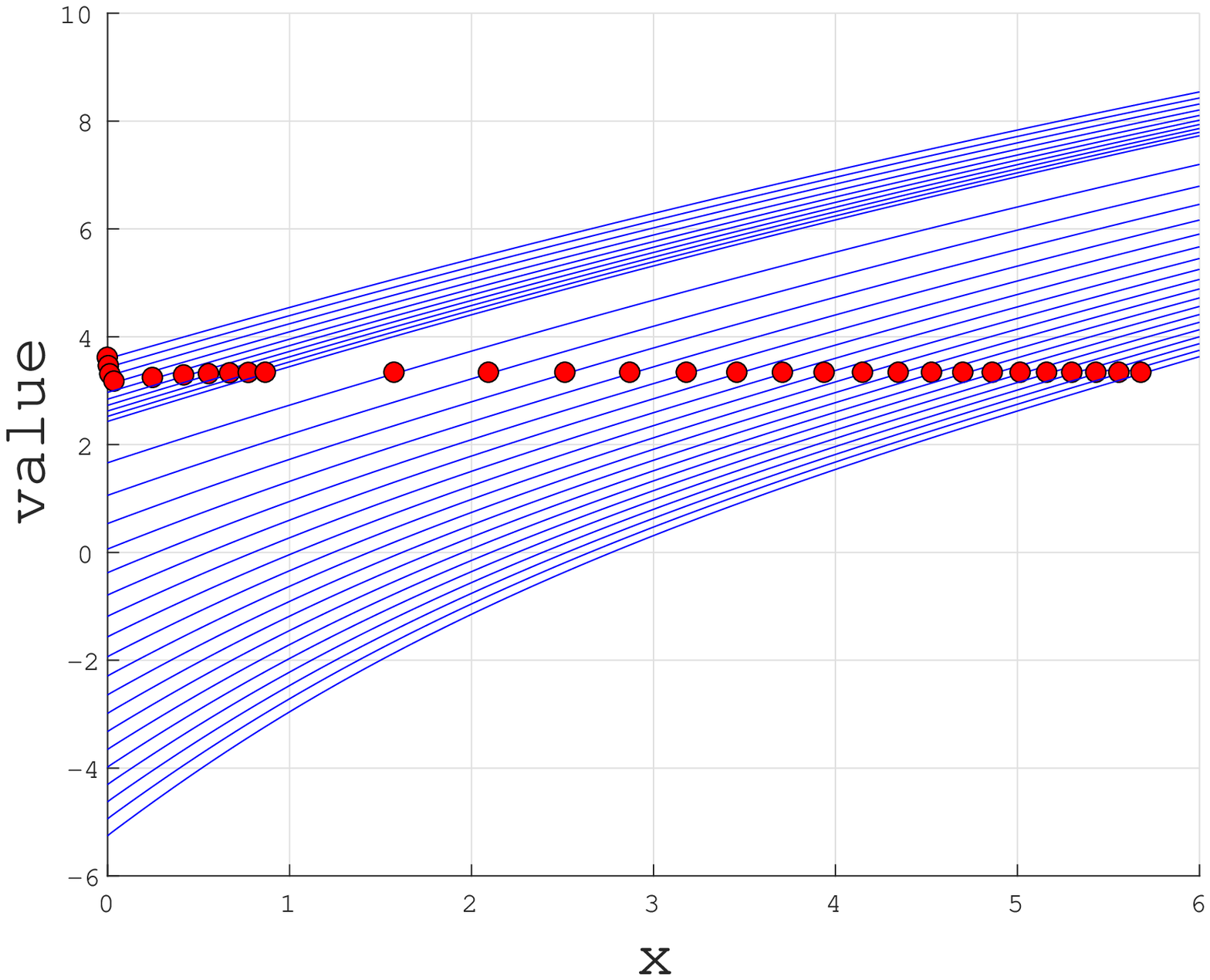} & \includegraphics[scale=0.30]{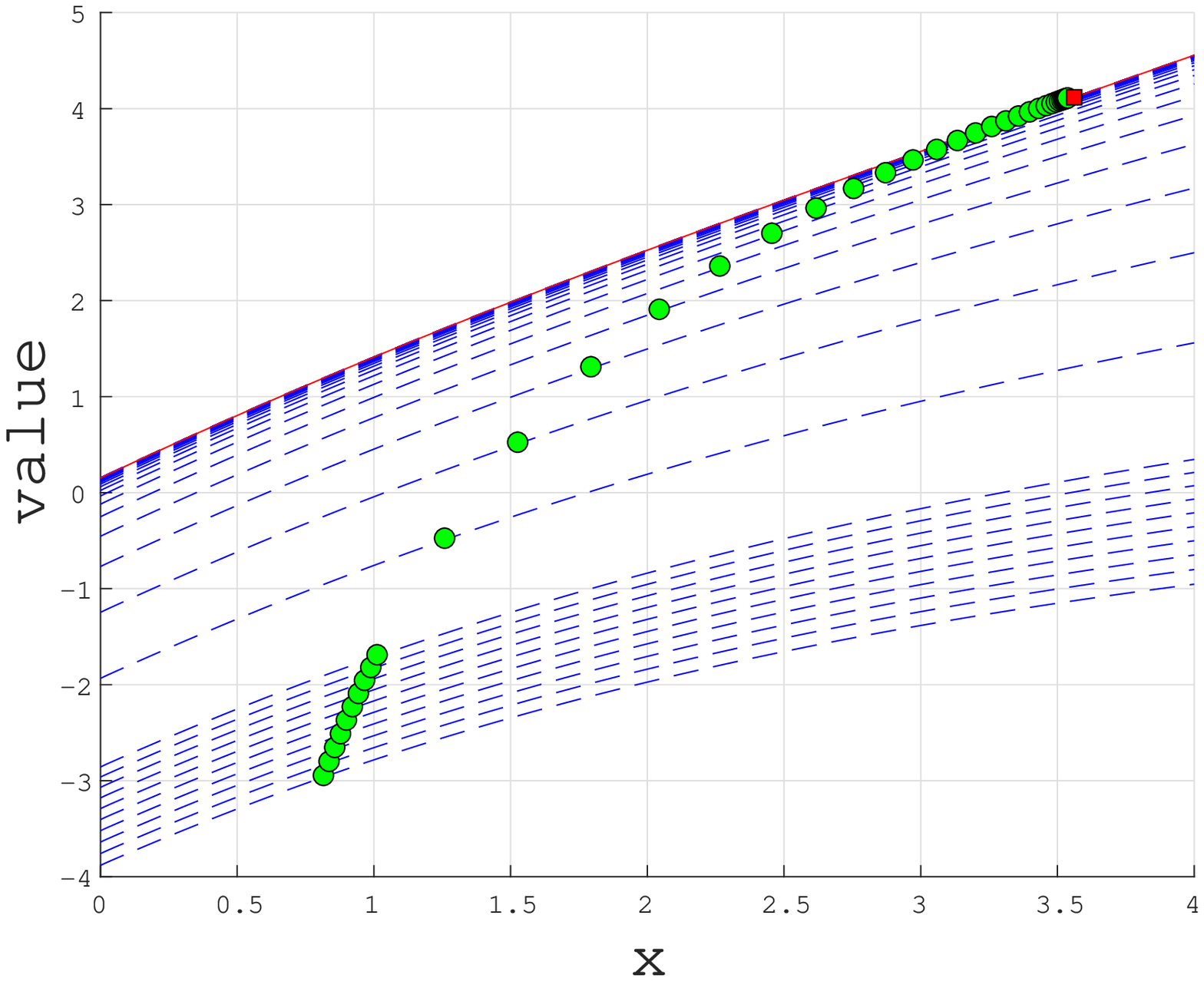}    \\
 sensitivity w.r.t.\ $\beta$ &  sensitivity w.r.t.\ $\delta$ \end{tabular}
\end{minipage}
\caption{(Left) plots of $x \mapsto v_{b^*}(x)$  for $\beta = 1.01$, $1.02$, $\ldots$, $1.09$, $1.1$,$1.2$, $\ldots$, $2.9$ and $3$, which correspond to curves from top to bottom. (Right) plots of $x \mapsto v_{b^*}(x)$ (dotted) $x \mapsto v_{b^*}(x)$  for $\delta = 0.01$, $0.02$, $\ldots$, $0.09$, $0.1$, $0.2$, $\ldots$, $2.9$ and $3$, which correspond to curves from bottom to top, along with the value function in \cite{1.} (solid). The circles indicate the points at $b^*$ and the square indicates the point at the optimal barrier in \cite{1.}.
}
\label{plot_sensitivity}
\end{center}
\end{figure}

%
%
\begin{singlespacing}
\begin{acknowledgements}
J. L. P\'erez is supported by CONACYT, project no. 241195. K. Yamazaki is in part supported by MEXT KAKENHI grant no. 17K05377. X. Yu is supported by Hong Kong Early Career Scheme under grant no. 25302116.
\end{acknowledgements}
\end{singlespacing}


\bibliographystyle{spmpsci_unsrt}

\end{document}